\documentclass[11pt]{techrep}
\usepackage{pat}
\usepackage{venturis}
\newtheorem{Definition}{Definition}[section]
\usepackage{verbatim}
\usepackage{subcaption}
\usepackage{enumitem}
\usepackage{amssymb}
\usepackage{wrapfig}
\usepackage{xcolor}
\usepackage{setspace}
\newcommand{\footremember}[2]{%
    \footnote{#2}
    \newcounter{#1}
    \setcounter{#1}{\value{footnote}}%
}
\newcommand{\footrecall}[1]{%
    \footnotemark[\value{#1}]%
} 

\newcommand{\tfour}{\ensuremath{\overrightarrow{\Theta_4}}}
\newcommand{\tkay}{\ensuremath{\overrightarrow{\Theta_k}}}
\newcommand{\tsix}{\ensuremath{\overrightarrow{\Theta_6}}}

\usepackage[ruled,vlined]{algorithm2e}

\usepackage{amsmath}

\usepackage{graphicx} % if you want to include jpeg or pdf pictures
\DeclareGraphicsExtensions{.pdf,.eps,.jpg,.png}

\newcommand{\pathf}[1]{\ensuremath{\mathcal{#1}}}
\newcommand{\T}{T}

\newcommand{\pathFG}[2]{\pathf{P}_{#1}\subP{#2}}
\newcommand{\pathFGA}[3]{\pathf{P}_{#1}^{#2}\subP{#3}}

\newcommand{\subP}[1]{\ensuremath{(#1)}}
\newcommand{\subA}[1]{\ensuremath{\langle#1\rangle}}
\newcommand{\pathF}[1]{\pathf{P}\subP{#1}}
\newcommand{\path}[1]{\pathf{P}\subA{#1}}

\newcommand{\ov}[1]{\bar{#1}}

\begin{document}

\title{On the Spanning and Routing Ratio of Directed Theta-Four\footnote{A preliminary version of this paper appeared in the proceedings of SODA 2019} \thanks{This work was partially supported by the Natural Sciences and Engineering Research Council of Canada (NSERC).}}

%\author{Prosenjit Bose\thanks{Carleton University} \and Jean-Lou De Carufel\thanks {University of Ottawa} \and Darryl Hill\thanks{Carleton University} \and Michiel Smid\thanks{Carleton University} }
% Document starts
\author{Prosenjit Bose\footremember{carleton}{Carleton University} \and Jean-Lou De Carufel\footremember{uofo} {University of Ottawa} \and Darryl Hill\footrecall{carleton} \and Michiel Smid\footrecall{carleton} }

\date{}
\maketitle

	\begin{abstract}		
		We present a routing algorithm for the directed $\Theta_4$-graph, here denoted as the \tfour-graph, that 
		computes a path between any two vertices $s$ and $t$ having length 
		at most $17$ times the Euclidean distance between $s$ and $t$. 
		To compute this path, at each step, the algorithm only uses knowledge of the location of the current
		vertex, its (at most four) outgoing edges, the destination vertex, 
		and one additional bit of information in order to determine the next edge to follow. This provides the first known online, local, competitive
		routing algorithm with constant routing ratio for the $\Theta_4$-graph, as well as 
		improving the best known upper bound on the spanning ratio of these graphs from $237$ to $17$. We also show that without this additional bit of information, the routing ratio increases to $\sqrt{290} \approx 17.03$.
	\end{abstract}

	\section{Introduction}

Finding a path in a graph is a fundamental problem in computer science. Typically, algorithms that compute paths in graphs have at their disposal knowledge of the whole graph. The problem of finding a path in a graph is more difficult in the online setting, when the routing algorithm must explore the graph as it attempts to find a path. Moreover, the situation is even more challenging if the routing algorithm only has a constant amount of working memory, i.e. it can only remember a constant size subgraph of the portion of the graph it has explored. Specifically, an online routing algorithm attempting to find a path from one vertex to another is called {\em local} if at each step, the only information it can use to make its forwarding decision is the location of the current vertex and its neighbouring vertices, plus a constant amount of additional information. 

For a routing algorithm $A$ and a given graph $G$ from a class $\mathcal{G}$ of graphs, let $\pathFGA{G}{A}{s,t}$ be the path in $G$ found by $A$ from $s$ to $t$. The class of graphs we focus on are a subclass of weighted geometric graphs. 
A weighted geometric graph $G=(P,E)$ is a graph whose vertex set is a set $P$ of points in the plane and a set $E$ of (directed or undirected) edges between pairs of points, where the weight of an edge $(p,q)$ is equal to the Euclidean distance $L_2(p,q)$ between its endpoints (i.e., distance in the $L_2$-metric). For a pair of vertices $s$ and $t$ in $P$, let $\pathFG{G}{s,t}$ be the shortest path from $s$ to $t$ in $G$, and let $L_2(\pathFG{G}{s,t})$ be the length of $\pathFG{G}{s,t}$ with respect to the $L_2$-metric, i.e., the sum of the lengths of the edges of $\pathFG{G}{s,t}$. The spanning ratio of a graph $G$ is the minimum value $c$ such that $L_2(\pathFG{G}{s,t})\leq c \cdot L_2(s,t)$ over all pairs of points $s$ and $t$ in $G$. A graph is called a \emph{$c$-spanner}, or just a \emph{spanner}, if its spanning ratio is at most some constant $c$. The routing ratio of a local online routing algorithm $A$ on $\mathcal{G}$ is the maximum value $c'$ such that $L_2(\pathFGA{G}{A}{s,t})\leq c'\cdot L_2(s,t)$ for all $G\in \mathcal{G}$ and all pairs $s$ and $t$ in $G$. When $c'$ is a constant, such an algorithm is called \emph{competitive} on the class $\mathcal{G}$. Note that the routing ratio on a class of graphs $\mathcal{G}$ is an upper bound on the spanning ratio of $\mathcal{G}$, since the routing ratio proves the existence of a bounded-length path. 
	
	\subsection{$\Theta$-graphs}
	
	Let $k \geq 3$ be an integer and for each $i$ with $0 \leq i < k$, 
	let $\mathcal{R}_i$ be the ray emanating from the origin that makes
	an angle of $2 \pi i /k$ in the counter-clockwise direction measured from the negative $y$-axis. Let $\mathcal{R}_k = \mathcal{R}_0$. The $\Theta_k$-graph of a given set $P$	of points is the directed graph that is obtained in the following way. The vertex set is the set $P$. Each vertex $v$ has at most $k$ 
	outgoing edges: For each $i$ with $0 \leq i < k$, let $\mathcal{R}_i^v$ be the 
	ray emanating from $v$ parallel to $\mathcal{R}_i$. Let $C_i^v$ be the cone consisting of all 
	points in the plane that are strictly between the rays $\mathcal{R}_i^v$
	and $\mathcal{R}_{i+1}^v$ or on $\mathcal{R}_{i+1}^v$.  
	If $C_i^v$ contains at least one point of $P \setminus \{v\}$, then 
	let $w_i$ be such a point whose perpendicular projection onto the 
	bisector of $C_i^v$ is closest to $v$ (where closest refers to the 
	Euclidean distance). Then the $\Theta_k$-graph contains the directed 
	edge $(v,w_i)$. See Figure~\ref{construction} for an example with $k=4$. While most of the literature discussed focuses on undirected $\Theta_k$-graphs, and thus $(v,w_i)$ becomes the undirected edge $\{v,w_i\}$, in this paper we will study routing in the directed setting. We will hereafter refer to directed $\Theta_k$-graphs using the $\tkay$ notation.

	$\Theta_k$-graphs were introduced independently by Keil and Gutwin~\cite{ch2keiltheta,ch2keil1992}, and Clarkson~\cite{ch2clarkson1987}. Both papers gave a spanning ratio of $1/(\cos\theta - \sin\theta)$, where $\theta=2\pi/k$ is the angle defined by the cones. Observe this gives a constant spanning ratio for $k\geq9$. Ruppert and Seidel~\cite{ch2ruppert} improved this to $1/(1-2\sin(\theta/2))$, which applies to $\Theta_k$-graphs with $k\geq 7$. Bose et al.~\cite{ch2morenotbetter} give a tight bound of $2$ for $k=6$. In the same paper are the current best bounds on the spanning ratio of a large range of values of $k$. More recently Bose et al.~\cite{conf91} showed that $\tsix$ has a spanning ratio of $7$. For $k=5$, Bose et al.~\cite{theta5arxiv} showed an upper bound on the spanning ratio of $\approx 5.70$. The previous bound by Bose et al.~\cite{ch2sander5} of $\approx 9.96$ also showed a lower bound of $\approx 3.78$. For $k=4$, Barba et al.~\cite{ch2barbat4} showed a spanning ratio of $\approx 237$, with a lower bound of $7$. For $k=3$, Aichholzer et al.~\cite{theta3} showed that $\Theta_3$ is connected, but Molla~\cite{ch2nawarphd} showed that there is no constant $c$ for which $\Theta_3$ is a $c$-spanner. 
	  
	\begin{figure}
		\centering
		\includegraphics[page = 3]{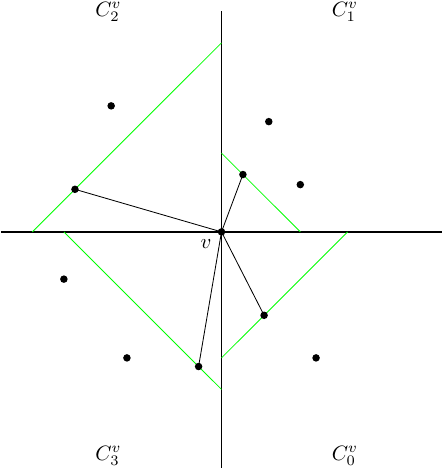}\caption{Neighbours of $v$ in the $\Theta_4$-graph. }\label{construction}
	\end{figure}
	
	\begin{comment}
	\begin{table}
			\begin{tabular}{ c c c}
			$\theta$ & Lower Bound & Upper Bound\\
			\hline\hline
			2,3 & $\infty$\cite{ch2nawarphd}& Not a spanner\cite{ch2nawarphd} \\
			4 & 7\cite{ch2barbat4}& $17$ (this paper) \\
			5 & $\approx 3.798$\cite{ch2sander5} & $\approx 9.96$\cite{ch2sander5} \\
			6 & 2\cite{ch2bonichon-theta} & 2\cite{ch2bonichon-theta}\\	
		\end{tabular}  
	\end{table}
\end{comment}
	
\subsection{Local Routing}

Local routing has been studied extensively in variants of the Delaunay graph as well as $\Theta_k$-graphs (see \cite{ch2onlinetriangulation, ch2bcd13, ch2ksu99, ch2chew, ch2chewtd, ch2bfrv15, ch2ruh09, ch2cddm12, ch2bdh16}). Also, more recently there has been interest in routing on \tkay-graphs~\cite{conf91}. There is an intimate connection between $\Theta_k$-graphs and variants of the Delaunay triangulation. For example, the existence of an edge in a $\Theta_k$-graph implies the existence of an empty triangle containing the edge (refer to Figure \ref{construction}). In a Delaunay triangulation, the existence of an edge implies the existence of an empty disk containing the edge (or some empty convex shape when considering variants of the Delaunay graph). Moreover, the Delaunay graph where the empty convex shape is an equilateral triangle (this is often referred to as the TD-Delaunay graph \cite{ch2bonichon-theta}) is a subgraph of the $\Theta_6$-graph.

Chew \cite{ch2chew} proved that the $L_1$-Delaunay graph has bounded spanning ratio by providing a local routing algorithm whose routing ratio is at most $\sqrt{10}$. 
Bose and Morin\cite{ch2bosem04} provided a competitive local routing algorithm that works on triangulations that have the \emph{diamond property}. This includes such graphs as the $L_2$-Delaunay triangulation, the greedy triangulation, and the minimum weight triangulation. Bose and Morin\cite{ch2onlinetriangulation} showed that there are no deterministic routing algorithms that work on any arbitrary graph. This implies that we must pair routing algorithms with particular classes of geometric graphs in order to route competitively. They also provided the first deterministic competitive routing algorithm on the $L_2$-Delaunay graphs. Bonichon et al.~\cite{ch2dtrouting} showed that we could route competitively on the $L_2$-Delaunay triangulation with a routing ratio of around $5.9$, using a generalization of Chew's~\cite{ch2chew} algorithm. This was the best known routing ratio for $L_2$-Delaunay triangulations until Bonichon et al.~\cite{ch2esa18} gave a new algorithm with a routing ratio of $3.56$, which is currently the best known. Bose et al.~\cite{ch2ht6} show that the half-$\theta_6$ graph, which is identical to the TD-Delaunay graph, has a routing ratio of $5/\sqrt{3}$, and this is shown to be tight. Since the spanning ratio of this graph is $2$, it is an example where a local routing algorithm cannot necessarily find the shortest path, and we see a separation between the routing and spanning ratios in this graph.

	For $\Theta_k$-graphs, there is a simple routing algorithm called \emph{cone-routing} or \emph{greedy-routing} that is competitive for $k\geq 7$. To route from a vertex $s$ to a vertex $t$, let $C_i^s$ be the cone of $s$ that contains $t$. Forward the packet from $s$ to its neighbour in $C_i^s$, and repeat this until the destination is reached. Let $\theta = 2\pi/k$, then for $k\geq 7$, Ruppert and Seidel~\cite{ch2ruppert} proved that cone routing gives a routing ratio of $1/(1-2\sin(\theta/2))$. Cone routing also has the advantage of only utilizing outgoing edges, so each vertex only needs to store the location of at most $k$ neighbours. That means these algorithms and results also apply to the \tkay-graphs. For $k<7$, cone-routing does not necessarily give a short path. In fact, Bose, De Carufel and Devillers~\cite{bosejocg20} showed that cone-routing has unbounded routing ratios for $k\leq 6$. However, for $k=6$, Bose et al.~\cite{ch2ht6} show that a different local online routing algorithm gives a routing ratio of $\sqrt{5}/2\approx2.89$. More recently, Bose et al.~\cite{conf91} give a local online routing algorithm for the \tsix-graph with a routing ratio of at most $14$. Prior to this work, there was no known competitive routing algorithm for $k=4$.  
			
	\subsection{Our Results}
	
	In this paper we improve the upper bound of the spanning ratio of \tfour-graphs (and, by extension, $\Theta_4$-graphs) from $237$ to $17$. We do this by providing a local online routing algorithm with a routing ratio of at most $17$. This is the first local routing algorithm for \tkay-graphs or $\Theta_k$-graphs for $k=4$, bringing us one step closer to obtaining competitive routing strategies on all \tkay- and $\Theta_k$-graphs with $k>3$. Our algorithm is slightly counter-intuitive since it sometimes takes a step in a cone that does not contain the destination. This is different from cone-routing that always takes a step in the cone that contains the destination. The algorithm is simple, and only uses knowledge of the destination vertex, the current vertex $v$, the neighbours of $v$, and one bit of additional information. If we forgo that bit of information, then the routing ratio increases to $\sqrt{290} \approx 17.03$. Additionally, like cone-routing, we route using only outgoing edges, so each vertex only needs to store the location of its at most 4 outgoing neighbours.  For the remainder of the paper, all edges $(u,v)$ are considered directed outgoing edges from $u$ to $v$, and when we refer to the neighbour $v$ of a vertex $u$ in a cone $C_i^u$, we are referring to the outgoing edge $(u,v)$ of $u$.

	The rest of the paper is organized as follows. Section \ref{three} gives the details of the routing algorithm that is used to navigate the $\tfour$-graph. In Section \ref{four} we analyze the length of the path found by the algorithm, and show an upper bound of $17$ on the routing ratio.  In Section \ref{five} we give an example of a path that shows this approach cannot do any better than a routing ratio of $17$. In Section \ref{four.five} we show how routing with only knowledge of the destination vertex increases the routing ratio to $\sqrt{290} \approx 17.03$. Section \ref{six} concludes the paper and gives some directions for future work.

\section{Algorithm}\label{three}

	%\begin{comment}
\begin{figure}
	\begin{subfigure}[b]{0.4\textwidth}
		\centering
		\includegraphics[page = 19, scale = 1]{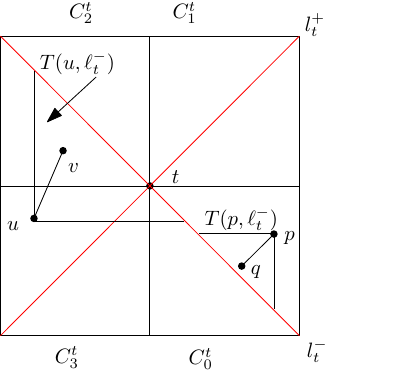}\caption{Two examples of sweeping steps towards $\ell_t^-$.}\label{sweep}
	\end{subfigure}
\hfill
	\begin{subfigure}[b]{0.4\textwidth}
		\centering
		\includegraphics[page = 20, scale = 1]{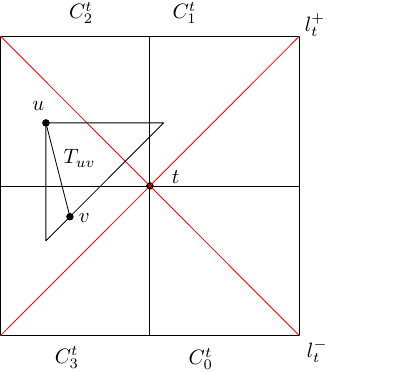}\caption{A greedy step towards $t$.\textcolor{white}{a little fill to align the pictures}}\label{greedy}
	\end{subfigure}
	\caption{}
\end{figure}
%\end{comment}

\begin{figure}
	\centering
	\includegraphics[page = 21, scale = 1]{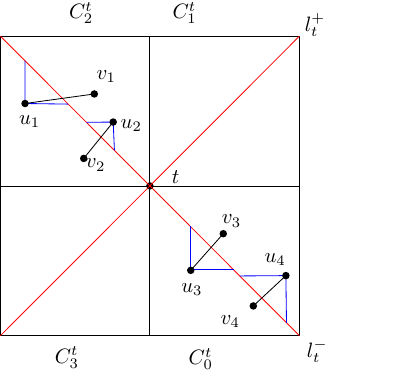}\caption{Vertices $u_1,u_2,u_3,$ and $u_4$ are all clean with respect to $\ell_t^-$.}\label{clear}
\end{figure}
	In this section, we present our 17-competitive local online routing algorithm on \tfour-graphs. We first introduce some concepts and notation related to the \tfour-graph. We then define the routing model, and finally we describe the routing algorithm in detail.
	\subsection{Preliminaries}
	Let $t$ be an arbitrary point in the plane, and let $\ell_t^-$ be the line through $t$ with slope $-1$. Similarly let $\ell_t^+$ be the line through $t$ with slope $1$. We refer to these as the \emph{diagonals} of $t$. Examples can be seen in Figures \ref{sweep}, \ref{greedy}, and \ref{clear}. To ease our analysis and avoid tedious tie-breaking, we make a general position assumption that no two vertices have the same $x$- or $y$-coordinates, and no two vertices lie on a common diagonal. Let $t$ and $u$ be arbitrary vertices and consider a diagonal of $t$. Without loss of generality, we consider the diagonal $\ell_t^-$ and assume that $u$ is in the half-plane below $\ell_t^-$. Let $\mathcal{R}_i^u$ and $\mathcal{R}_{i+1}^u$ be the rays emanating from $u$ that intersect $\ell_t^-$. Recall that $\mathcal{R}_i^u$ and $\mathcal{R}_{i+1}^u$ delineate the cone $C_i^u$. Let the triangle $T(u,\ell_t^-)$ be the intersection of the halfplane of $\ell_t^-$ containing $u$ and the cone $C_i^u$. We say that $C_i^u$ \emph{faces} $\ell_t^-$. If $T(u,\ell_t^-)$ is empty of vertices (not including $t$), then we say that $u$ is \emph{clean with respect to $\ell_t^-$}. See Figure \ref{clear}. If the diagonal we are referring to is clear from the context, we simply say that $u$ is \emph{clean}. If $u$ is not clean with respect to $\ell_t^-$, then let $v$ be the vertex in $T(u,\ell_t^-)$ for which $(u,v)$ is an edge in the $\tfour$-graph. We will refer to following the edge from $u$ to $v$ as \emph{taking a sweeping step towards} $\ell_t^-$. (See Figure~\ref{sweep}.) Let $i$ be the index such that the vertex $t$ is in the cone $C_i^u$.	Let $v$ be the vertex in $C_i^u$ for which $(u,v)$ is an edge in the $\tfour$-graph. We will refer to following the edge from $u$ to $v$ as \emph{taking a greedy step towards} $t$. (See Figure~\ref{greedy}.)  Note that when routing with respect to $\ell_t^-$ (respectively $\ell_t^+$) and the current vertex $v$ is in $C_3^t$ or $C_1^t$ (respectively $C_0^t$ or $C_2^t$), a greedy step towards $t$ and a sweeping step towards $\ell_t^-$ (respectively $\ell_t^+$) are the same. However, by our definition of \emph{clean} and to disambiguate the analysis, this step is defined as a sweeping step. 
	
	\subsection{Routing Model}
	
	An online local routing algorithm takes as input $u,t,N(u),m$ where $u$ is the current vertex, $t$ is the target vertex, $N(u)$ are the (1-hop) neighbours of $u$ in $G$, and $m=\{0,1\}^*$ is a bitstring of memory. The algorithm returns a vertex $v\in N(u)$ on the path from $u$ to $t$ and updates $m$ if necessary. The maximum length of $m$ over all steps of the algorithm represents the memory requirements of the algorithm. If the maximum length of $m$ is $0$ we say it is a \emph{memoryless} algorithm. That is, the algorithm does not require any knowledge of the previous vertices, including the start vertex. The strongest version of this routing model uses unlimited memory, while the weakest version of this routing model is memoryless. 
	
	The local $Greedy/Sweep$ routing algorithm that we define here has two versions, one that is a memoryless with a routing ratio of $\sqrt{290}\approx 17.03$ and one that uses $1$ bit of memory and has a routing ratio of $17$. We describe the $1$ bit routing in this section. We show how to get memoryless routing with a small modification to our algorithm in Section \ref{four.five}.

\subsection{The $Greedy/Sweep$ Algorithm}
 
	We now define the $1$-bit version of the $Greedy/Sweep$ algorithm. Let $s$ be the source vertex and $t$ the target vertex. The algorithm first chooses a diagonal of $t$ as follows: If $s \in C_0^t\cup C_2^t$, the algorithm chooses $\ell_t^-$, otherwise it chooses $\ell_t^+$. Intuitively the algorithm chooses the diagonal of $t$ ``closer" to $s$, and $d$ encodes this information about the location of $s$. Our routing algorithm is then denoted by $Greedy/Sweep(u,t,N(u),d)$ where $u$ is the current vertex, $t$ is the target vertex, $N(u)$ are the ($1$-hop) neighbours of $u$ and $d\in\{0,1\}$ is the memory representing the chosen diagonal of $t$. That is, $d=0$ corresponds to $\ell_t^-$ and $d=1$ corresponds to $\ell_t^+$. For the purposes of simplification, we will express $d$ as the diagonal directly, that is, we will assume that $d\in\{\ell_t^-,\ell_t^+\}$. 
	
	The $Greedy/Sweep(u,t,N(u),d)$ algorithm uses three other ``helper" algorithms to determine the output vertex $v$, which we will define here. Let $v \in N(u)$ be the vertex in the cone of $u$ that faces $d$. Note this implies that $(u,v)$ is an edge in $\tfour$. Then $Clean(u,t,N(u),d)$ returns True if $(u,v)$ crosses $d$ and False otherwise.  $Sweep(u,t,N(u),d)$ returns the vertex $v$ such that $(u,v)$ is the edge in the cone of $u$ facing $\ell_t^-$. $Greedy(u,t,N(u))$ returns the vertex $v$ such that $(u,v)$ is the edge in the cone of $u$ that contains $t$. Then the algorithm $Greedy/Sweep(u,t,N(u),d)$ is described in Algorithm \ref{algo}.
	
	\begin{algorithm}
\SetAlgoLined
\SetKwInOut{Input}{Input}
\Input{$u$ is the current vertex;\newline
$t$ is the target vertex;\newline
$N(u)$ are the neighbours of $u$;\newline
$d\in\{\ell_t^-,\ell_t^+\}$ is the chosen diagonal of $t$;}
\setstretch{1.2}
\vspace{0.1cm}\hrule\vspace{0.1cm}
 \lIf{$u=t$}{END}\vspace{0.1cm}
  \uIf{$Clean(u,t,N(u),d)$}{
   return $Greedy(u,t, N(u))$\;
   }\uElse{
   return $Sweep(u,t,N(u),d)$\;
  }
 \caption{$Greedy/Sweep(u,t,N(u),d)$}\label{algo}
\end{algorithm}

 Figure \ref{fig:example1} gives an example of a path from $s$ to $t$ computed by Algorithm \ref{algo}.

\section{Analysis}\label{four}

\begin{figure}   
 		\centering
 		\includegraphics[page =4, width= 8cm]{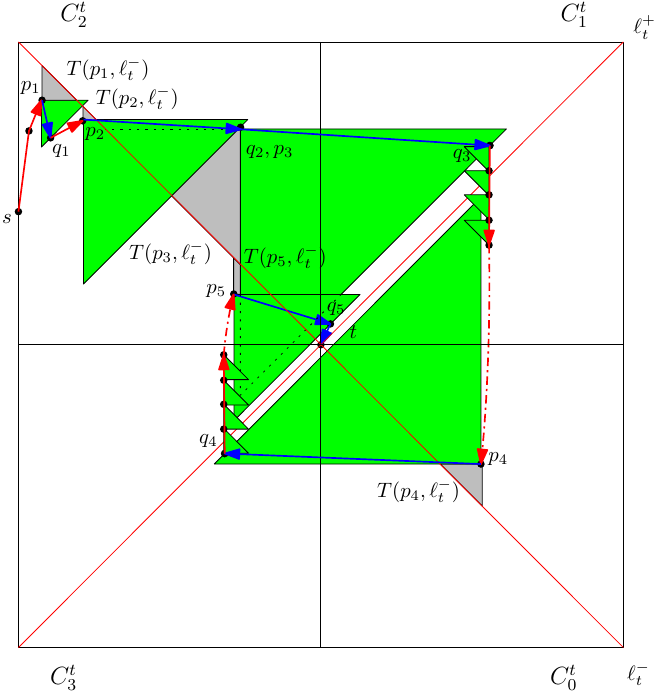}\caption{An example of the algorithm showing canonical triangles and cleaned triangles. Blue lines are greedy steps and red lines are sweeping steps.}\label{fig:example1}
 	\end{figure}
  
 	\begin{figure}
 		\centering
 		\includegraphics[page =5, width= 8cm]{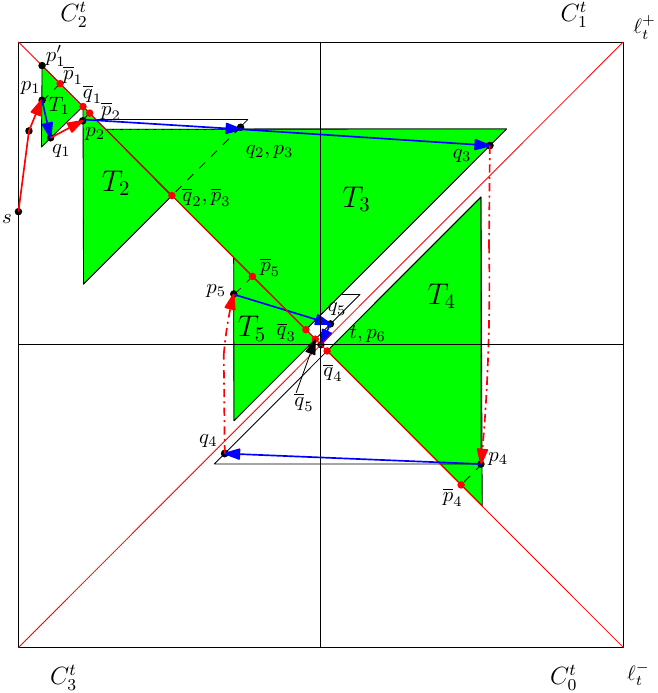}\caption{The bounding triangles. Blue lines are greedy steps towards $t$, red lines are sweeping steps towards $\ell_t^-$.}\label{fig:example2}

 	%\caption{An example path found by the algorithm and the associated triangles. }\label{example}
 \end{figure}

In this section, we prove that our routing algorithm terminates and that it has a routing ratio of 17. Without loss of generality, we assume that $s$ is in $C_2^t$ under $\ell_t^-$. Thus $\ell_t^-$ is the closest diagonal of $t$ to $s$. For two arbitrary points $u$ and $v$, let $d_x(u,v)$ and $d_y(u,v)$ be the distance between them along the $x$-axis and $y$-axis respectively. Let $L_1(u,v)$ be the $L_1$ distance between $u$ and $v$ (i.e., $L_1(u,v)= d_x(u,v)+d_y(u,v)$), and let $L_\infty(u,v)$ be the $L_\infty$ distance from $u$ to $v$ (i.e., $L_\infty(u,v) = \max\{d_x(u,v),d_y(u,v)\}$). To simplify our analysis, most of our intermediate measurements will be in the $L_1$-metric. In the final analysis we will express the length in the $L_2$-metric. Let $\pathF{s,t}$ be the sequence of directed edges produced by our algorithm. For vertices $u$ and $v$ in $\pathF{s,t}$, with $u$ occurring before $v$, let $\path{u,v}$ be the subpath of $\pathF{s,t}$ from $u$ to $v$.

We divide the area around $t$ into \emph{quadrants}. The \emph{Northern} quadrant is the area above $\ell_t^-$ and $\ell_t^+$, while the \emph{Southern} quadrant is the area below $\ell_t^-$ and $\ell_t^+$. The \emph{Western} quadrant is the area to the left of $\ell_t^-$ and $\ell_t^+$, while the \emph{Eastern} quadrant is the area to the right of $\ell_t^-$ and $\ell_t^+$. 
%See Figure \ref{fig:quadrant}.
\begin{comment}
\begin{figure}
	\centering
	\includegraphics[page =2]{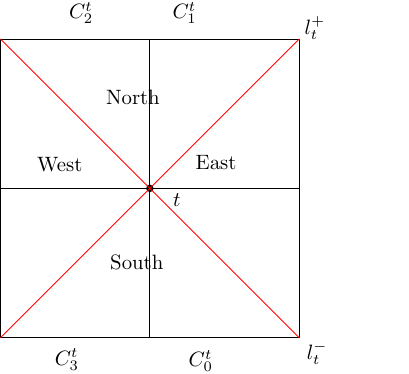}\caption{The cones of $t$ are separated by black lines, and the quadrants of $t$ are separated by red lines.}\label{fig:quadrant}
\end{figure}\qquad
\begin{figure}
	\centering
	\includegraphics[page = 1]{pics/ineighbour.pdf}
	\caption{The canonical triangle $T_{uv}$ in blue.}\label{canonical} refer to the little diagrams.
\end{figure}
\end{comment}

 We note the following about the path found by the algorithm:
\begin{lemma}\label{linfty}
	Let $u$ and $v$ be two consecutive vertices on $\pathF{s,t}$. Then $L_\infty(u,t)> L_\infty(v,t)$.
\end{lemma}

\begin{proof}
	
	Note that we can without loss of generality  assume that $u$ is a vertex in the Western quadrant and that we are routing with respect to $\ell_t^-$ by rotating the point set around $t$ or by flipping the point set along $\ell_t^-$. Since $u$ is below $\ell_t^-$ and above $\ell_t^+$, $L_\infty(u,t) = d_x(u,t)$. If $v$ is in the Western quadrant, observe that for both a greedy step and a sweeping step, $d_x(u,t)>d_x(v,t)$ since we assume no two vertices have the same $x$-coordinate. This implies that $L_\infty(u,t)> L_\infty(v,t)$, as required.
	
	Assume $v$ is not in the Western quadrant. For $v$ to be in the Eastern quadrant, $(u,v)$ must cross both diagonals of $t$. Observe that a greedy step does not cross $\ell_t^+$, while a sweeping step does not cross $\ell_t^-$. Thus $v$ can only be in the Northern or Southern quadrant, and $L_\infty(v,t)= d_y(v,t)$. Observe that if $v$ is in the Southern quadrant, then $(u,v)$ was a sweeping step, and if $v$ is in the Northern quadrant, then $(u,v)$ was a greedy step. In both cases, $d_y(u,t)>d_y(v,t)$ since we assume no two vertices have the same $y$-coordinate. Thus  $L_\infty(u,t)= d_x(u,t)>d_y(u,t)>d_y(v,t)=L_\infty(v,t)$, as required.  See Figs. \ref{sweep} and \ref{greedy}.
  \end{proof}
  
Since there are a finite number of vertices, this leads to the following corollary.

\begin{corollary} 
	 The $Greedy/Sweep$ algorithm terminates, i.e., it reaches $t$. 
\end{corollary}

 Let $((p_1,q_1),(p_2,q_2),$ $...,(p_{m-1},q_{m-1}))$ be the sequence of edges produced by greedy steps of the algorithm, with $t=p_m$. A \emph{phase} $f_i$ of the algorithm refers to the path from a vertex $p_i$ to a vertex $p_{i+1}$ consisting of $(p_i,q_i)+\path{q_i,p_{i+1}}$, for $1\leq i <m$. That is, a phase consists of a single greedy step followed by a (possibly empty) sequence of sweeping steps. Note that the first phase is preceded by a (possibly empty) sequence of sweeping steps from $s$ to $p_1$. Let $L_2(f_i)= L_2(p_i,q_i)+L_2(\path{q_i,p_{i+1}})$ represent the length of phase $f_i$. Then observe that $\pathF{s,t} = \path{s,p_1}+\sum_{i=1}^{m-1} f_i$, where the $+$ operator on paths is concatenation of the paths. 
 Note that if each vertex on $\path{u,v}$ is in the same cone $i, 0\leq i\leq 3$ of all preceding vertices, then $\path{u,v}$ is $x$- and $y$-monotone, and $L_2(\path{u,v})\leq L_1(u,v)$. This implies that $L_2(\path{q_i,p_{i+1}})\leq L_1(q_i,p_{i+1})$, for all $1\leq i<m-1$.
 
 Let $v$ be the neighbour of an arbitrary vertex $u$ in the cone $C_i^u$. Let the \emph{canonical triangle} $T_{uv}$ be the triangle formed by the boundaries of $C_i^u$ and the line through $v$ perpendicular to the bisector of $C_i^u$. Note that the existence of $(u,v)$ guarantees that $T_{uv}$ is empty of vertices in its interior. See Figure \ref{greedy}.

\begin{figure}
	\begin{subfigure}[b]{0.4\textwidth}
		\includegraphics[page = 5, scale = 1]{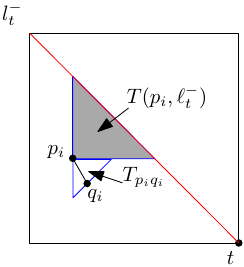}
	\end{subfigure}
	\hfill
	\begin{subfigure}[b]{0.4\textwidth}
		\includegraphics[page = 6]{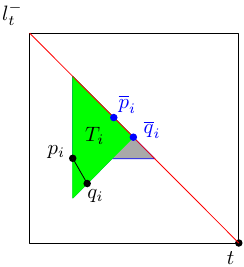}
	\end{subfigure}
	\caption{$T_{p_iq_i}$ does not intersect $\ell_t^-$.}\label{can1.1}
\end{figure}

\begin{Definition}
	Consider the edge $(p_i,q_i)$ of $f_i$. If $p_i$ is in the Northern or Southern quadrant, then let $\mathcal{L}$ be the horizontal line through $p_i$, otherwise $\mathcal{L}$ is the vertical line through $p_i$. Let the \emph{bounding triangle} $\T_i$ be the triangle formed by the lines $\ell_t^-$, $\mathcal{L}$, and $\ell_{q_i}^+$. See Figs. \ref{can1.1}, \ref{can1.2}, and \ref{can1.3}.
\end{Definition}

\begin{lemma}\label{empty}
	The bounding triangle $\T_i$ of the greedy edge $(p_i,q_i)$ of phase $f_i$ is empty of vertices.
\end{lemma}

\begin{proof}
	Since $(p_i,q_i)$ is a greedy step, $p_i$ is clean with respect to $\ell_t^-$, and $T(p_i,\ell_t^-)$ and $T_{p_iq_i}$ are both empty of vertices. Observe that $T_i$ lies completely in one of the half-planes of $\ell_t^-$. If $T_{p_iq_i}$ does not intersect $\ell_t^-$, then $T_{p_iq_i} \subseteq T_i$ and $T(p_i,\ell_t^-) \nsubseteq T_i$. See Figure \ref{can1.1}. If $T_{p_iq_i}$ does intersect $\ell_t^-$, then observe that $T(p_i,\ell_t^-) \subseteq T_i$ and $T_{p_iq_i} \nsubseteq T_i$. In this case, $q_i$ can be on the same side of $\ell_t^-$ as $p_i$, and thus lie on $T_i$ (Figure \ref{can1.2}), or it can be on the opposite side of $\ell_t^-$, and not lie on $T_i$ (Figure \ref{can1.3}). In all cases observe that $\T_i \subseteq T_{p_iq_i}\cup T(p_i,\ell_t^-)$, and thus $\T_i$ is empty of vertices.
 \end{proof}

\begin{figure} 
	\begin{subfigure}[b]{0.4\textwidth}
		\includegraphics[page = 3]{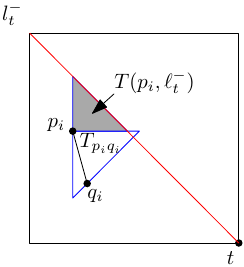}
	\end{subfigure}
\hfill
	\begin{subfigure}[b]{0.4\textwidth}
		\includegraphics[page = 4]{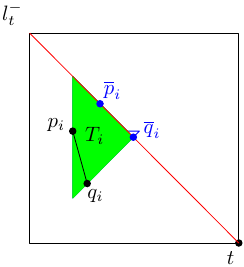}
	\end{subfigure}
	\caption{$T_{p_iq_i}$ intersects $\ell_t^-$, and $q_i$ lies on $T_i$.}\label{can1.2}
\end{figure}

\begin{figure}
	\begin{subfigure}[b]{0.4\textwidth}
		\includegraphics[page = 7]{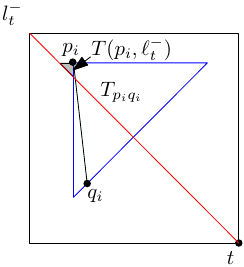}
	\end{subfigure}
	\hfill
	\begin{subfigure}[b]{0.4\textwidth}
		\includegraphics[page = 8]{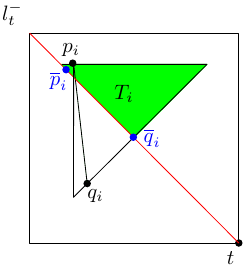}
	\end{subfigure}
	\caption{$T_{p_iq_i}$ intersects $\ell_t^-$, and $q_i$ does not lie on $T_i$.}\label{can1.3}
\end{figure}

Notice that a bounding triangle $\T_i$ cannot be on both sides of $\ell_t^-$ by construction, and cannot be on both sides of $\ell_t^+$ since that would imply that $t$ is within $T_{p_iq_i}$. This implies that a bounding triangle $\T_i$ can only intersect the interior of a single quadrant. 

  Lemma \ref{linfty} has strong implications about the positions of bounding triangles relative to one another in the same quadrant. For a vertex $p$, let $\ov{p}$ be the intersection of $\ell_t^-$ and $\ell_p^+$, i.e, $\ov{p}$ is the intersection of the positive diagonal of $p$ and the negative diagonal of $t$.

\begin{lemma}\label{orderedtriangles}
	If $\T_i$ and $\T_j$ are two bounding triangles in the same quadrant, then $\ov{p}_j\ov{q}_j$ and $\ov{p}_i\ov{q}_i$ are disjoint segments on $\ell_t^-$. %Furthermore, if $i<j$, then $\ov{p_i}\ov{q_i}$ is farther from $t$ than $\ov{p_i}\ov{q_i}$.  
\end{lemma}
 
\begin{proof}
	
	Without loss of generality, assume that $i<j$. Note that for a point $\ov{v}$ lying on $\ell_t^-$, $L_1(\ov{v},t)=\sqrt{2}\cdot L_2(\ov{v},t)$. That is, the $L_1$- and $L_2$-distances are proportional. Then
	Lemma \ref{orderedtriangles} is true if $L_1(\ov{p}_i,t)>L_1(\ov{q}_i,t)>L_1(\ov{p}_j,t)>L_1(\ov{q}_j,t)$ is true. Assume without loss of generality $\T_i$ and $\T_j$ are in the Western quadrant. See Figure \ref{potential2}.
	 Note that $L_1(\ov{p}_i,t)-L_1(\ov{q}_i,t)= L_1(\ov{p}_i,\ov{q}_i)= L_1(p_i,q_i)$, since $q_i$ and $t$ are in the same cone of $p_i$. Thus $L_1(\ov{p}_i,t)>L_1(\ov{q}_i,t)$ and $L_1(\ov{p}_j,t)>L_1(\ov{q}_j,t)$ are true. What remains to be shown is that $L_1(\ov{q}_i,t)>L_1(\ov{p}_j,t)$.
	Lemma \ref{linfty} implies $L_\infty(p_i,t)>L_\infty(p_j,t)$, and both points are in the Western quadrant (by the definition of bounding triangle), thus $p_j$ cannot be left of $p_i$. This, and the fact that $\T_i$ is empty, implies $p_j$ must be below $\ell_{q_{i}}^+$, which implies $\ell_{p_j}^+$ is below $\ell_{q_i}^+$, which implies $L_1(\ov{q}_i,t)>L_1(\ov{p}_j,t)$. 
\end{proof}

 Figure \ref{potential2} shows two consecutive bounding triangles in the Western quadrant, and the associated segments $\ov{p}_i\ov{q}_i$ and $\ov{p}_{i+1}\ov{q}_{i+1}$.

 Let $p_1'$ be the vertical projection of $p_1$ onto $\ell_t^-$. Then the following inequality is true.

\begin{corollary}\label{lemma49}
	 $\sum_{i=1}^{m-1} L_1(\ov{p}_i,\ov{q}_i) \leq 4\cdot L_1(p_1',t)$.
\end{corollary}

\begin{proof}
	Lemma \ref{linfty} implies that $L_1(p_1',t)>L_1(\ov{p}_i,t)$ for all $1\leq i \leq m-1$. This combined with Lemma \ref{orderedtriangles} and the fact that there are four quadrants implies the lemma. 
  \end{proof}

\begin{figure}
	\centering
	\includegraphics[page = 22, width=9cm]{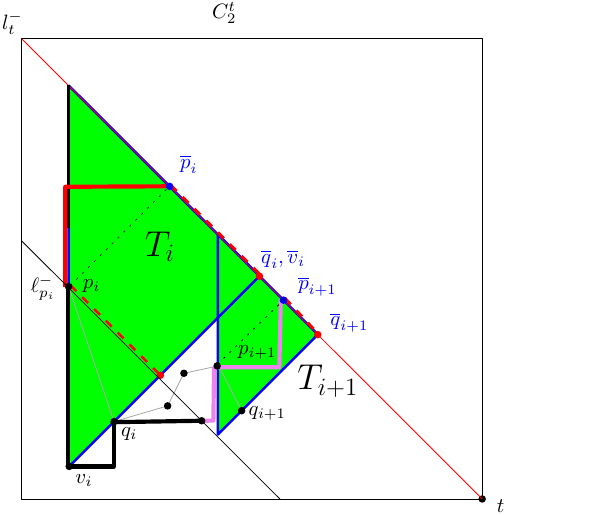}\caption{The red path from $p_i$ to $\ov{p}_i$ is the same length as the violet path. Thus $\Phi(p_i,p_{i+1})$ is equal to the length of the black path. }\label{potential2}
\end{figure}
   
\begin{figure}
	\centering
	\includegraphics[page = 23, width = 8.5cm]{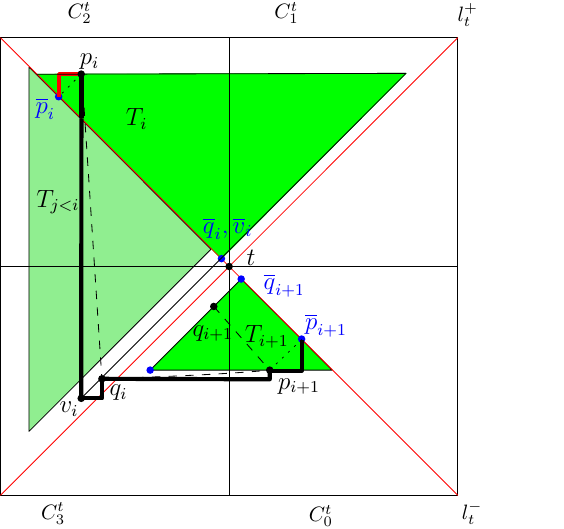}\caption{$L_1(q_i,\ov{p}_{i+1})<L_1(p_i,q_i)$, since $p_i$ is in $C_2^{q_i}$ and on the opposite side of $\ell_t^-$ as $q_i$.}\label{potential1}
\end{figure}

\begin{lemma}\label{thislemma}
	$L_2(f_i)\leq L_1(p_i,q_i) + L_1(q_i, p_{i+1}).$
\end{lemma}

\begin{proof}
	This follows from the fact that $(p_i,q_i)$ is an edge, and $\path{q_i,p_{i+1}}$ is $x$- and $y$-monotone. 
  \end{proof}

Each of the bounding triangles $\T_i$ are associated with the segment $\ov{p}_i\ov{q}_i$ and the phase $f_i$. A natural approach is to try to bound the length of the phase $L_2(f_i)$ by the length of the segment $\ov{p}_i\ov{q}_i$. Unfortunately this approach does not quite work since the length of $\ov{p}_i,\ov{q}_i$ is proportional to the length of $(p_i,q_i)$, but the length of the sequence of sweeping steps of $f_i$ can be unbounded with respect to the length of $(p_i,q_i)$.  However, a relatively simple potential function reassigns the lengths of sweeping steps to either the previous phase or the next phase so that these lengths are proportional to the length of $(p_i,q_i)$. 
We define the potential function $\Phi(p_i,p_{i+1})  = L_1(p_i,q_i) + L_1(q_i,\ov{p}_{i+1})-L_1(p_i,\ov{p}_i)$ for all $1\leq i <m-1$. 

\begin{lemma}\label{lemma412}
	$\Phi(p_i,p_{i+1})\leq 2\cdot L_1(\ov{p}_i,\ov{q}_i)$.
\end{lemma}

\begin{proof}
	
	Without loss of generality, assume that $q_i$ is in the Western quadrant. Since greedy steps can only originate in $C_2^t$ and $C_0^t$, and a greedy step cannot cross $\ell_t^+$, $p_i$ is in $C_2^t$ in either the Western or Northern quadrant (Figs. \ref{potential2} and \ref{potential1} respectively). Let $v_i$ be the bottommost point of $T_{p_iq_i}$. Since $v_i$ is on $\ell_{q_i}^+$, we have $\ov{v_i}=\ov{q_i}$. Since both $p_i$ and $\ov{p}_i$ are on $\ell_{p_i}^+$, and both $v_i$ and $\ov{v}_i$ are on $\ell_{v_i}^+$, we have $L_1(p_i,v_i)=L_1(\ov{p}_i,\ov{v_i}) =L_1(\ov{p}_i,\ov{q_i})$. Thus it is enough to prove that $\Phi(p_i,p_{i+1})\leq 2\cdot L_1(p_i,v_i)$. Observe that  $L_1(p_i,v_i)+L_1(v_i,q_i)\geq L_1(p_i,q_i)$ by the triangle inequality. 
	
	If $p_i$ is in the Western quadrant, we have that $L_1(v_i,\ov{p}_i) = L_1(v_i,p_i)+L_1(p_i,\ov{p}_i)$.
	Since both $\ov{p}_i$ and $\ov{p}_{i+1}$ lie in $C_1^{v_i}$ on $\ell_t^-$, we have also have that $L_1(v_i,\ov{p}_i) = L_1(v_i,\ov{p}_{i+1})$. Thus
	
	\begin{align*}
		\Phi(p_i,p_{i+1})& = L_1(p_i,q_i) + L_1(q_i,\ov{p}_{i+1})-L_1(p_i,\ov{p}_i)\\
		&\leq L_1(p_i,v_i)+L_1(v_i,\ov{p}_{i+1})-L_1(p_i,\ov{p}_i)\\
		&= L_1(p_i,v_i)+L_1(v_i,\ov{p}_{i})-L_1(p_i,\ov{p}_i)\\
		&= 2\cdot L_1(p_i,v_i)
	\end{align*}
	
	as required. Otherwise $p_i$ is in the Northern quadrant. Observe that $p_i$ and $\ov{p}_{i+1}$ are both in $C_1^{v_i}$, but $p_i$ is above $\ell_t^-$ while $\ov{p}_{i+1}$ is on it, thus $L_1(p_i,v_i)> L_1(v_i, \ov{p}_{i+1})$. Thus 
	
	\begin{align*}
	\Phi(p_i,p_{i+1})& = L_1(p_i,q_i) + L_1(q_i,\ov{p}_{i+1})-L_1(p_i,\ov{p}_i)\\
	&\leq L_1(p_i,v_i)+L_1(v_i,\ov{p}_{i+1})\\
	&\leq 2\cdot L_1(p_i,v_i)
	\end{align*}
	
	as required.
	  \end{proof}
 
We can now prove the main theorem. 

\begin{theorem}\label{maintheorem}
	The path produced by Algorithm \ref{algo} has length at most $17\cdot L_2(s,t)$.
\end{theorem}

\begin{proof}
	Recall that $t=p_{m}$. Thus $L_1(p_{m}, \ov{p}_{m})=0$, and 
	
	\begin{equation}
	\sum_{i=1}^{m-1} (L_1(p_{i+1} ,\ov{p}_{i+1})-L_1(p_i,\ov{p}_i)) = -L_1(p_1,\ov{p}_1).\label{line}
	\end{equation}
	
	Since $p_1$ is in $C_1^s$, and $\ov{p}_1$ is in $C_1^{p_1}$, we have that $\ov{p}_1$ is in $C_1^s$. Since we assume that $\ell_t^-$ is the closest diagonal to $s$, that gives us
	
	\begin{equation}
	L_1(s,\ov{p}_1)\leq L_\infty(s,t) \leq L_2(s,t)\label{line3}. 
	\end{equation}
	
	Additionally, since $p_1'$ is a point on $\ell_t^-$, we have $L_1(p_1',t)= 2L_\infty(p_1',t)$. Observe that $L_\infty(s,t)> L_\infty(p_1,t)=L_\infty(p_1',t)$, thus $L_1(p_1',t) \leq 2L_\infty(s,t)$, and $L_\infty(s,t)\leq L_2(s,t)$. That gives us
	
	\begin{equation}
		L_1(p_1',t) \leq 2\cdot L_\infty(s,t)\leq 2 \cdot L_2(s,t). \label{line2}
	\end{equation}

	 Thus $L_2(\pathF{s,t})$ is equal to

	\begin{align*}
	&L_2(\path{s,p_1})+\sum_{i=1}^{m-1} L_2(f_i)\nonumber\\
	\leq &~L_1(s,p_1)+\sum_{i=1}^{m-1} (L_1(p_i,q_i)+L_1(q_i,p_{i+1}))\text{ (Lemma \ref{thislemma})}\\
	=&~ L_1(s,p_1)+L_1(p_{1},\ov{p}_{1}) -L_1(p_{1},\ov{p}_{1})+\sum_{i=1}^{m-1} \left(L_1(p_i,q_i)+L_1(q_i,p_{i+1})\right)\nonumber \\
	= &~L_1(s,\ov{p}_{1})+\sum_{i=1}^{m-1}(L_1(p_{i+1}, \ov{p}_{i+1})-L_1(p_i,\ov{p}_i))+\sum_{i=1}^{m-1} (L_1(p_i,q_i)+L_1(q_i,p_{i+1}))\\
	&\hspace{5.7cm}\text{ (above is by \eqref{line})}\\%\label{thatlabel}\\
	= &~L_1(s,\ov{p}_{1})+\sum_{i=1}^{m-1} (L_1(p_i,q_i)+L_1(q_i,\ov{p}_{i+1})-L_1(p_i,\ov{p}_i))\\
	= &~L_1(s,\ov{p}_1)+\sum_{i=1}^{m-1} {\Phi(p_i,p_{i+1})} \\
	\leq &~ L_1(s,\ov{p}_1) + 2
	\sum_{i=1}^{m-1} L_1(\ov{p}_i,\ov{q}_i)\hspace{2.15cm}\text{ (Lemma \ref{lemma412})}\\%\label{theotherlabel}\\
	\leq &~ L_1(s,\ov{p}_1)+ 8\cdot L_1(p_1',t)\hspace\fill\hspace{2.6cm}\text{ (Corollary \ref{lemma49})}\\%\label{lastlabel}\\
	\leq &~ L_2(s,t)+ 16\cdot L_2(s,t)\hspace{2.5cm}\text{ (by \eqref{line3} and \eqref{line2})}\\
	\leq &~ 17\cdot L_2(s,t)
	\end{align*}

	\begin{comment}
	\begin{eqnarray*}
	&&L_2(\path{s,p_1})+\sum_{i=1}^{m-1} L_2(f_i)\nonumber\\
	& \leq & L_1(s,p_1)+\sum_{i=1}^{m-1} (L_1(p_i,q_i)+L_1(q_i,p_{i+1}))\text{ (Lemma \ref{thislemma})}\\
	&=& L_1(s,p_1)+L_1(p_{1},\ov{p}_{1}) -L_1(p_{1},\ov{p}_{1})+\sum_{i=1}^{m-1} \left(L_1(p_i,q_i)+L_1(q_i,p_{i+1})\right)\nonumber \\
	&=&L_1(s,\ov{p}_{1})+\sum_{i=1}^{m-1}(L_1(p_{i+1}, \ov{p}_{i+1})-L_1(p_i,\ov{p}_i))+\sum_{i=1}^{m-1} (L_1(p_i,q_i)+L_1(q_i,p_{i+1}))\\
	&&\hspace{5.7cm}\text{ (above is by \eqref{line})}\\%\label{thatlabel}\\
	&=&L_1(s,\ov{p}_{1})+\sum_{i=1}^{m-1} (L_1(p_i,q_i)+L_1(q_i,\ov{p}_{i+1})-L_1(p_i,\ov{p}_i))\\
	&=&L_1(s,\ov{p}_1)+\sum_{i=1}^{m-1} {\Phi(p_i,p_{i+1})} \\
	&\leq& L_1(s,\ov{p}_1) + 2
	\sum_{i=1}^{m-1} L_1(\ov{p}_i,\ov{q}_i)\hspace{2.15cm}\text{ (Lemma \ref{lemma412})}\\%\label{theotherlabel}\\
	&\leq& L_1(s,\ov{p}_1)+ 8\cdot L_1(p_1',t)\hspace\fill\hspace{2.6cm}\text{ (Corollary \ref{lemma49})}\\%\label{lastlabel}\\
	&\leq& L_2(s,t)+ 16\cdot L_2(s,t)\hspace{2.5cm}\text{ (by \eqref{line3} and \eqref{line2})}\\
	&\leq& 17\cdot L_2(s,t)
	\end{eqnarray*}
	\end{comment}
	as required. 
  \end{proof}

\section{Lower bound}\label{five}
 
\begin{figure}
	%\begin{subfigure}[b]{0.45\textwidth}
	\centering
	\includegraphics[page = 1, width= 8.4cm]{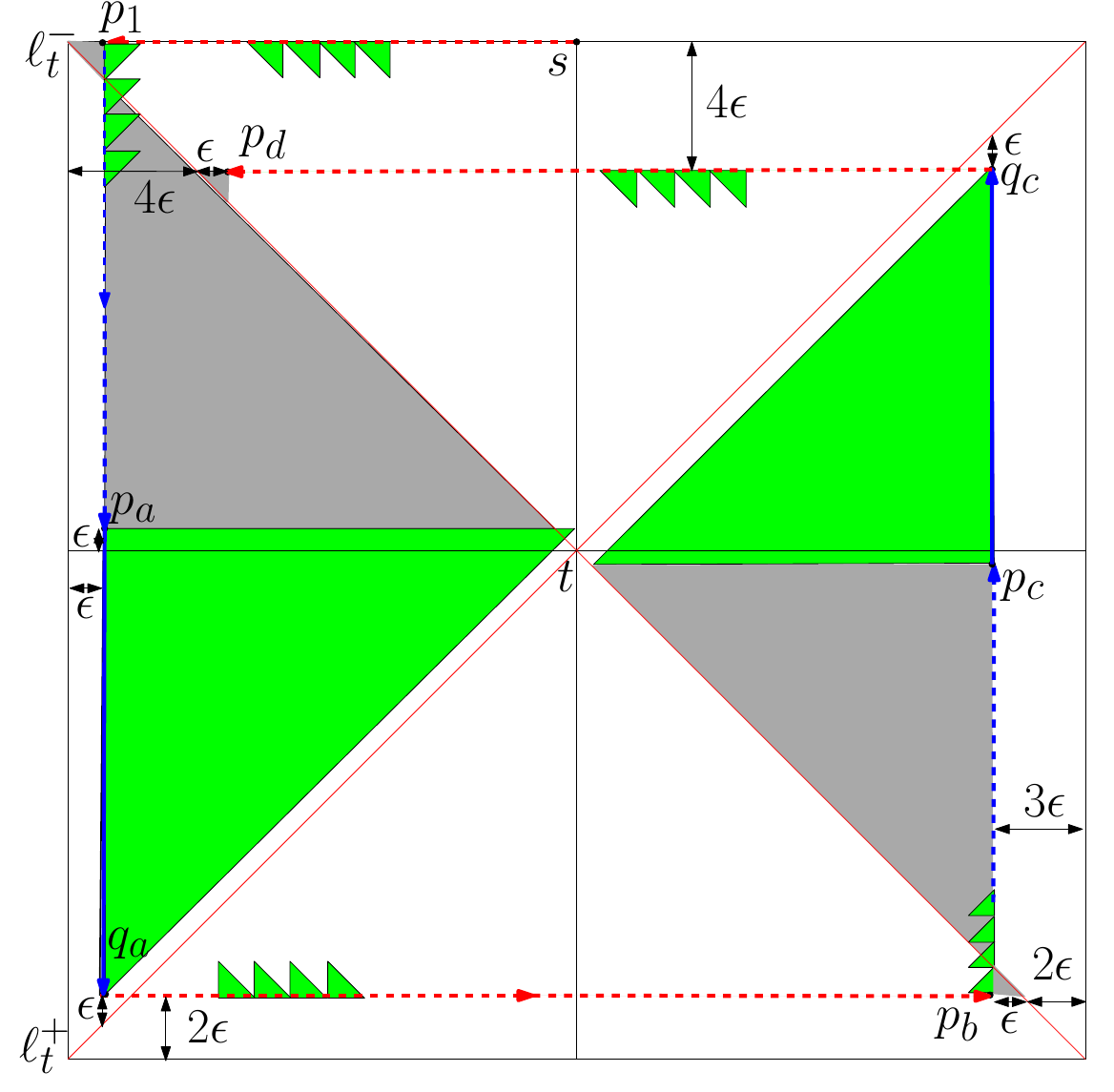}\caption{ The path from $s$ to $p_d$. Blue lines are greedy steps, red are sweeping steps towards $\ell_t^-$. The grey and green regions are empty of vertices.}\label{lowerbound1}
	%\end{figure}
	%\begin{figure}
	%\end{subfigure}
	%\begin{subfigure}[b]{0.45\textwidth}
	%	\centering
	\includegraphics[page = 2, width= 8.4cm]{pics/lowerbound3.pdf}\caption{The path from $p_d$ to $p_e$. }\label{lowerbound2}
	%\end{subfigure}
\end{figure}
\begin{figure}
	%\begin{subfigure}[b]{0.45\textwidth}
	\centering
	\includegraphics[page = 3, width= 8.4cm]{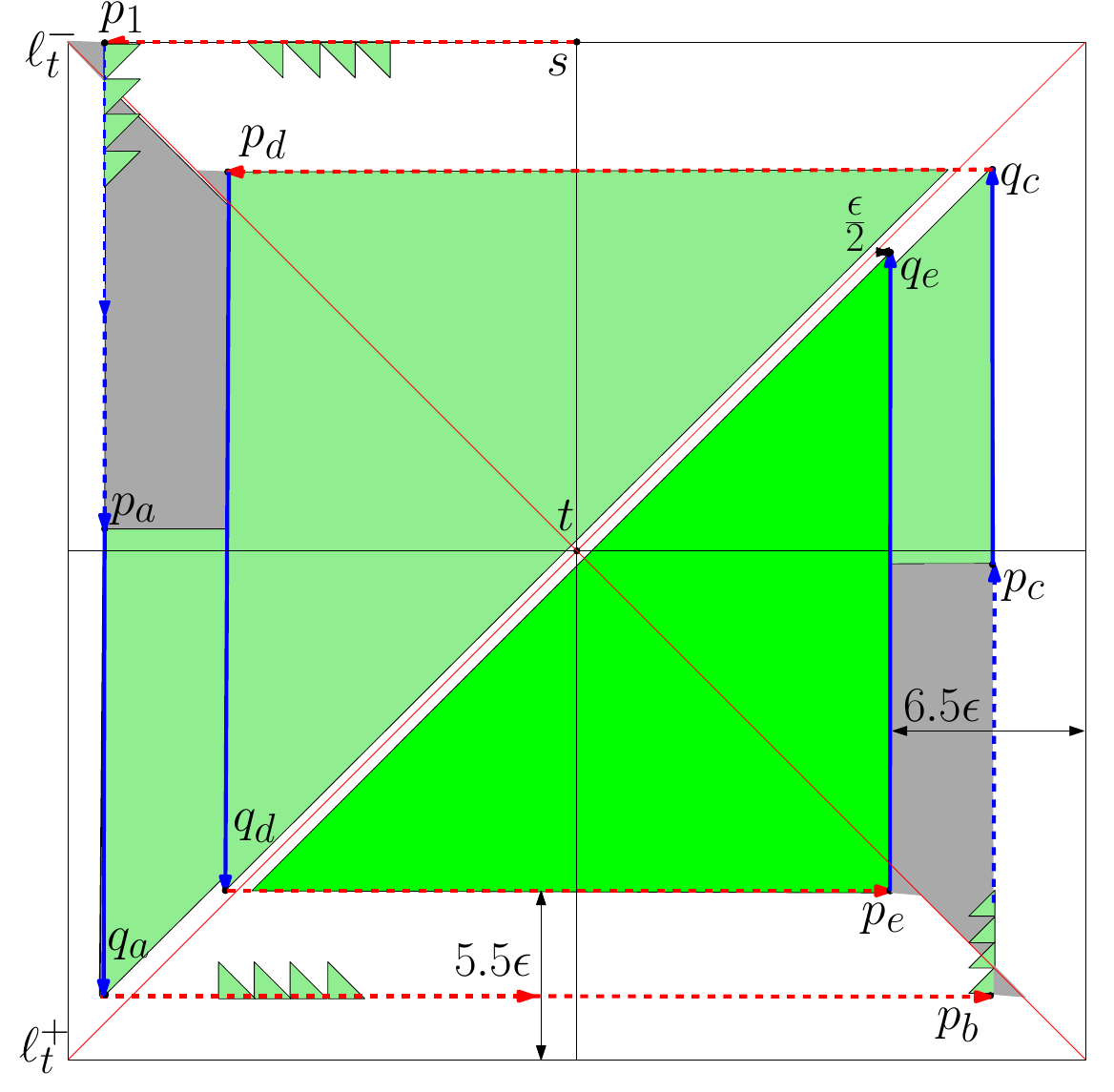}\caption{The last greedy edge $(p_e,q_e)$.\vspace{0.44cm}}\label{lowerbound3}
	%\end{figure}
	%\begin{figure}
	
	\includegraphics[page = 4, width= 8.4cm]{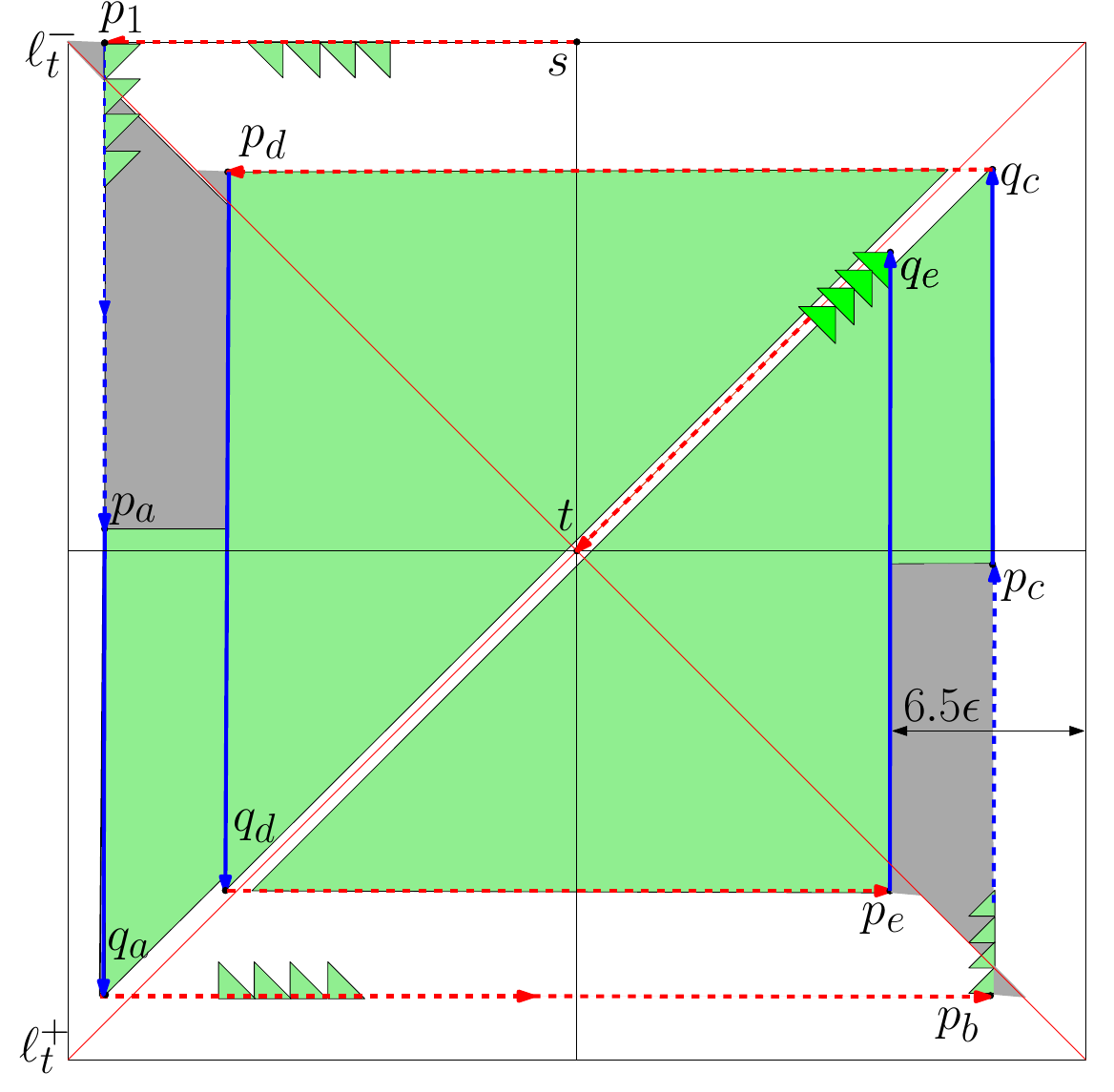}\caption{The final path from $q_e$ to $t$.}\label{lowerbound4}
	%\end{subfigure}
\end{figure}

In this section, we show that our analysis of the routing ratio of Algorithm~\ref{algo} is tight: We will construct a set of points, together with two vertices $s$ and $t$, such that the routing ratio of Algorithm~\ref{algo} is arbitrarily close to $17$. The construction is illustrated in Figs.~\ref{lowerbound1}, \ref{lowerbound2}, \ref{lowerbound3}, and \ref{lowerbound4}. We forgo our general position assumption in order to make the demonstration of the lower bound simpler. For this particular example, if a vertex is on the boundary between two cones or two quadrants, we say that that vertex is in the \emph{counter-clockwise} of the two cones or quadrants. Let $\epsilon>0$ be an arbitrarily small number. Let $1<a<b<c<d<e$ be (not necessarily consecutive) integers. 
 
Let $t$ be at coordinates $(0,0)$. Let $s$ be at coordinates $(0,1)$. Vertex $p_1$ is at $(-1+\epsilon,1)$. See Figure \ref{lowerbound1}. Place a sequence of vertices directly left of $s$ at coordinates $(-\epsilon,1),(-2\epsilon,1),(-3\epsilon,1)..., (-1+\epsilon, 1)=p_1$. This implies that $s$ is not clean, so we take sweeping steps along this sequence of vertices (red dashed line) until we reach $p_1$. The path from $s$ to $p_1$ has length $(1-\epsilon)$.

Directly below $p_1$ (and thus in $C_0^{p_1}$) there is a sequence of vertices at coordinates $(-1+\epsilon, 1-\epsilon),(-1+\epsilon,1-2\epsilon),..., (-1+\epsilon, \epsilon)=p_a$, all of which are clean. Since $p_1$ is clean with respect to $\ell_t^-$, we take greedy steps along this sequence to $p_a$ (blue dashed line). Directly below $p_a$ is vertex $q_a=(-1+\epsilon, -1+2\epsilon)$. Vertex $p_a$ is clean, so the next greedy step takes us to $q_a$ (blue edge). The path from $p_1$ to $q_a$ has length $2-2\epsilon$. 

To the right of $q_a$ is a sequence of vertices at coordinates $(-1+2\epsilon, -1+2\epsilon),(-1+3\epsilon, -1+2\epsilon), ..., (1-3\epsilon, -1+2\epsilon)=p_b$. The path from $q_a$ to $p_b$ has length $2-4\epsilon$.

There is a sequence of vertices directly above $p_b$ at coordinates $(1-3\epsilon, -1+3\epsilon), (1-3\epsilon, -1+4\epsilon), ... ,(1-3\epsilon, -\epsilon)=p_c$, all of which are clean. Thus we proceed along these vertices in a sequence of greedy steps from $p_b$ to $p_c$ (blue dashed path). From $p_c$ we take a greedy step to $q_c=(1-3\epsilon,1-4\epsilon)$ (blue edge). The path from $p_b$ to $q_c$ has length $2-6\epsilon$. 

There is a sequence of vertices left of $q_c$ at coordinates $(1-4\epsilon,1-4\epsilon),(1-5\epsilon,1-4\epsilon),...,(-1+5\epsilon,1-4\epsilon)=p_d$. The path from $q_c$ to $p_d$ has length $2-8\epsilon$. 

In Figure \ref{lowerbound2} we take a greedy step from $p_d$ to $q_d=(-1+5\epsilon, -1 +5.5\epsilon)$ (blue edge). This edge has length $2-9.5\epsilon$. 

To the right of $q_d$ is a sequence of vertices at coordinates $(-1+6\epsilon, -1 +5.5\epsilon),(-1+7\epsilon, -1 +5.5\epsilon),...,(1-6\epsilon, -1 +5.5\epsilon)=p_e$. The path from $q_d$ to $p_e$ has length $2-11.5\epsilon$. 

In Figure \ref{lowerbound3} we take a greedy step from $p_e$ to $q_e=(1-6.5\epsilon, 1-7\epsilon)$ (blue edge). The edge $(p_e,q_e)$ has length $2 - 12.5\epsilon$. 

In Figs. \ref{lowerbound4} and \ref{lowerbound5} there are a sequence of vertices at $(1-7\epsilon, 1-7\epsilon),(1-7\epsilon+\epsilon', 1-7.5\epsilon),(1-7.5\epsilon, 1-7.5\epsilon),(1-7.5\epsilon+\epsilon', 1-8\epsilon)...(0,0)=t$. We will define $\epsilon'$ in a moment. A sequence of clearing steps takes us from $q_e$ to $t$ along these vertices. Let $\delta+1$ be the number of horizontal edges in this sequence, and let $\gamma$ be the number of edges with a vertical component.  Let $\epsilon'=\epsilon/\delta$. Observe that $d_x(q_e,t)= 1-6.5\epsilon = (\delta+1)\cdot\epsilon/2$. The first horizontal edge has length $\epsilon/2$, and the remaining $\delta$ horizontal edges have length $\epsilon/2 - \epsilon'$. Thus the total length of the horizontal edges is $1-6.5\epsilon - \delta\epsilon' = 1-6.5\epsilon - \epsilon = 1-7.5\epsilon$. 

Observe that $d_y(q_e,t) = 1-7\epsilon = \gamma\cdot\epsilon/2$. Each of the $\gamma$ vertical edges has length $>\epsilon/2$, since their vertical distance is $\epsilon/2$ and they are skewed from vertical, thus the total length of the vertical steps is at least $1-7\epsilon$. Thus the path from $q_e$ to $t$ has length at least $2-14.5\epsilon$. 
     
The total length of these paths is at least $L_2(s,p_1)+L_2(p_1,q_a) +L_2(q_a, p_b) +L_2(p_b,q_c)+L_2(q_c,p_d)+L_2(p_d,q_d)+L_2(q_d,p_e)+L_2(p_e,q_e)+ d_x(q_e,t)-\epsilon +d_y(q_e,t) = 1-\epsilon+
  2-2\epsilon +
  2-4\epsilon +
  2-6\epsilon +
  2-8\epsilon +
  2-9.5\epsilon +
  2-11.5\epsilon +
  2-12.5\epsilon +
  2-14.5\epsilon = 17 - 69\epsilon$. Since $L_2(s,t)=1$, by letting $\epsilon$ tend to $0$ we can make the path arbitrarily close to $17\cdot L_2(s,t)$. This gives us the following theorem.\\
  
  \begin{theorem}
  	There exists a set of points such that the distance travelled by Algorithm \ref{algo} is at least $17-\epsilon$ for any $\epsilon>0$.
  \end{theorem}
\begin{figure}
	\centering
 \includegraphics[width = 8.4cm]{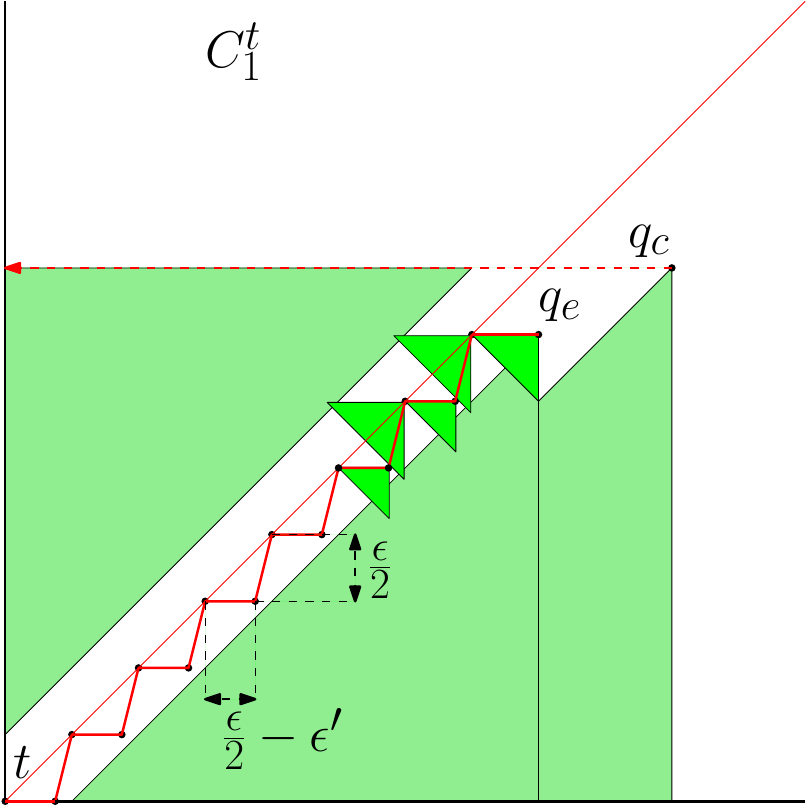}\caption{Details of the final path from $q_e$ to $t$.}\label{lowerbound5}
\end{figure}
\section{Removing the Diagonal-Bit}\label{four.five}
 
 The algorithm, as presented in Section~\ref{three} uses one single bit to remember the diagonal of the destination $t$ that is closest to the start vertex $s$. In this section, we show that without this single bit, the routing ratio increases to $\sqrt{290} < 17.03$.
 
 Our modification of Algorithm~\ref{algo} is to ``hard code" the diagonal we route with respect to into our main and helper functions. For instance, the helper algorithms $Clean(u,t,N(u))$ and $Sweep(u,t,N(u))$ no longer require the $1$-bit $d$ as input. Instead, they always make their decisions with respect to $\ell_t^-$ regardless of the position of $s$. The routing algorithm $Greedy/Sweep(u,t,N(u))$ is now memoryless. It does not have a parameter $d$ and it returns only a vertex $v \in N(u)$.  The changes in the analysis are in Inequality~\eqref{line3}, which becomes
 
 \begin{equation}
 L_1(s,\ov{p}_1)\leq d_y(s,t) + L_\infty(s,t), \label{line4}
 \end{equation}
 
 and in ~\eqref{line2}, which becomes
 \begin{equation}
 L_1(p_1',t)\leq 2\cdot L_\infty(s,t). \label{line5}
 \end{equation}
 
 If we replace \eqref{line4} and \eqref{line5} by \eqref{line3} and \eqref{line2} respectively in our proof of Theorem \ref{maintheorem} we get $L_2(\path{s,t})\leq d_y(s,t) + 17\cdot L_\infty(s,t)$. Let $\gamma= (d_y(s,t) + 17\cdot L_\infty(s,t))/L_2(s,t)$. The routing ratio is thus the maximum of $\gamma$. Let $u$ be the point at $(d_x(s),d_y(t))$, and let $\theta = \angle uts$. We can rewrite $\gamma$ as $\sin\theta + 17\cdot \cos\theta = \sqrt{17^2 + 1^2}\cdot\sin(\theta + \arctan(17))$ for $0 \leq \theta \leq \pi/4$. This is maximized at $\theta = \arctan(\frac{1}{17})$ with a value of $\sqrt{290}$. Thus we have the following theorem.
 
 \begin{theorem}
  With no bits of memory, and using a fixed diagonal $\ell_t^-$, Algorithm \ref{algo} outputs a path from $s$ to $t$ with length at most $\sqrt{290}\cdot L_2(s,t)$.
 \end{theorem}

 If we refer to the lower bound proof in Section~\ref{five}, we can adjust it to this new bound by moving $s$ to the right until $st$ forms an angle of $\arctan(\frac{1}{17})$ with the positive $y$-axis. Thus, in this case, we can get arbitrarily close to $\sqrt{290}$.
    
\section{Conclusion}\label{six}

We have presented a simple online local routing algorithm for $\tfour$-graphs that achieves a routing ratio of $17$ using knowledge of the destination and one bit of information, and $\sqrt{290}<17.03$ using only knowledge of the destination. 
Although we have presented the first such algorithm on $\tfour$-graphs and also improved the spanning ratio of \tfour- and $\Theta_4$-graphs from 237 down to 17, we conjecture that this upper bound both on the routing ratio and spanning ratio is not tight. Given that  7~\cite{ch2barbat4} is the best known lower bound for the spanning ratio of $\Theta_4$, the actual spanning ratio remains unknown.

\begin{comment}

\end{comment}

\bibliographystyle{plain}
\bibliography{newbib}

\end{document}